\documentclass[12pt, reqno]{amsart}
\usepackage{amssymb}
\usepackage{amsmath}
\allowdisplaybreaks[1]
\usepackage{graphicx,color}
\usepackage{wrapfig,framed}
\usepackage[height=20cm, width=15cm, hmarginratio={1:1}]{geometry}
\usepackage[colorlinks=true,citecolor=red,linkcolor=blue]{hyperref}
\usepackage{bm,enumerate}

\newcommand{\im} {\mbox {Im}\hskip 0.5truemm}

\theoremstyle{plain}
\newtheorem{theorem}{Theorem}
\newtheorem{prop}{Proposition}
\newtheorem{lemma}{Lemma}
\newtheorem{cor}{Corollary}

\theoremstyle{definition}
\newtheorem{defi}{Definition}
\newtheorem*{remark}{Remark}
\newtheorem{example}{Example}

\newcommand{\beq}{\begin{equation}}
\newcommand{\eeq}{\end{equation}}
\newcommand{\nn}{\nonumber}

\newcommand{\CC}{{\mathbb C}}
\newcommand{\R}{{\mathcal R}}

\newcommand{\bt}{{\bf t}}

\newcommand{\p}{\partial}
\newcommand{\g}{{\mathfrak g}}

\def\={\; = \;}
\def\+{\, + \,}
\def\:={\; := \; }

\def \L {\mathcal L}

\newcommand{\ad}{\mathrm{ad}}

\newcommand{\Ker}{\mathrm{Ker}}

\newcommand{\fb}{\mathfrak{b}}
\newcommand{\fn}{\mathfrak{n}}

\newcommand{\mf}[1]{\mathfrak{#1}}
\newcommand{\mb}[1]{\mathbb{#1}}
\newcommand{\mc}[1]{\mathcal{#1}}

\DeclareMathOperator{\Der}{Der}

\def\wbigoplus{\mathop{\widehat{\bigoplus}}\limits}

\definecolor{light}{gray}{.9}

\begin{document}

\title[From tau-structure to integrable hierarchy]
{Remarks on intersection numbers and integrable hierarchies. II. 
Tau-structure}

\author{Daniele Valeri}
\address{Dipartimento di Matematica \& INFN, Sapienza Universit\`a di Roma,
P.le Aldo Moro 5, 00185 Roma, Italy}
\email{daniele.valeri@uniroma1.it}

\author{Di Yang}
\address{School of Mathematical Sciences, USTC, Hefei 230026, P.R. China}
\email{diyang@ustc.edu.cn}

\begin{abstract}
For systems of evolutionary partial differential equations the tau-structure is
an important notion which originated from the deep relation between integrable systems and quantum field theories.
We show that, under a certain non-degeneracy condition, existence of a
tau-structure implies integrability. 
As an example, we apply this principle to provide a new proof of the integrability of
the Drinfeld--Sokolov hierarchy associated to an arbitrary Kac--Moody algebra and 
a choice of a vertex of its Dynkin diagram.
\end{abstract}

\maketitle

\setcounter{tocdepth}{1}

\tableofcontents

\section{Introduction}
The study of the deep relation between topology  
of the Deligne--Mumford moduli space $\overline{\mathcal{M}}_{g,n}$ 
of stable algebraic curves~\cite{DM}
and integrable hierarchies started from 
Witten's conjecture~\cite{Witten} (first proved by Kontsevich~\cite{Kontsevich}), which states that  
the exponential of the generating series of certain intersection numbers on $\overline{\mathcal{M}}_{g,n}$ 
is a particular tau-function for the Korteweg--de Vries (KdV) hierarchy.
Later, generalizations were made through different 
directions (see for examples \cite{DZ-norm, LRZ, YZpreprint}).
Due to this deep relation, on one hand, the theory of 
integrable systems provides 
powerful tools and methods for understanding 
 topology of the moduli spaces, e.g. 
in computing the intersection numbers on 
$\overline{\mathcal{M}}_{g,n}$; 
on the other hand, the intersection numbers provide rich mathematical structures~\cite{Dubrovin1, KM, Manin} 
 that help to develop the theory of integrable systems.  
For more details, see for examples~\cite{BDY0, BDY, BPS, Dubrovin1, DLYZ1, DYZ, DZ-norm}.

This paper belongs to a paper series, where we plan to discuss several topics in mathematics that have important origins from   
the study of the relationships between intersection numbers and integrable hierarchies. The main topic for this paper is 
the notion of {\it tau-structure}. This notion is closely related to a more famous one, {\it tau-function},
which provides the tool that allows to present the relationships in an elegant and comprehensible way.

In the theory of integrable systems, the concept of tau-function
appeared in 1970s--1980s (see for examples~\cite{Hirota, Sato}). 
(Actually, theta-functions in the theory of Riemann surfaces,   
and partition functions in combinatorics, quantum field theory, statistical physics 
and random matrices 
are also closely related to tau-functions.) 
In this approach the tau-function is provided by a solution to Hirota bilinear equations, and, geometrically, it corresponds to a point in an infinite dimensional manifold.

In this paper, we consider the notion of tau-function 
using the viewpoint from the Dubrovin--Zhang theory~\cite{DZ-norm} (cf.~\cite{Dubrovin1, DLYZ1, DYZ, Witten}). 
In this context,
certain polynomials associated to the tau-function (or, equivalently, to the corresponding system of partial differential equations), called the two-point correlation functions,
turn out to play an important role in understanding the geometry behind the tau-function 
itself \cite{BDY0, BPS, Dickey, DLYZ1, DYZ, DZ-norm, EYY, Witten}. Following Dubrovin and Zhang,
we refer to the two-point correlation functions as the tau-structure; actually, 
our terminology on tau-structure (see Definition~\ref{def:tau} below and see also~\cite{DYZ, YZpreprint}) is
more general than the one in~\cite{DZ-norm}. Recall that, in the relationships between intersection numbers and integrable hierarchies, 
the former should correspond to Taylor coefficients of the {\it logarithm} of some tau-function of an integrable hierarchy, and this is one of the main sources that motivates to define the tau-structure as the family of second 
derivatives of the logarithm of a tau-function (cf.~\eqref{eq:tau}).

The goal of this paper is to understand better the interplay between tau-structure and integrability
of systems of partial differential equations.
We show in Theorem~\ref{tauint} that under a certain non-degeneracy assumption (see \eqref{mondegen}) the existence of a tau-structure implies integrability.
As an application, using the tau-structure constructed in \cite{DVY}, we give a novel proof of the integrability of the 
Drinfeld--Sokolov (DS) hierarchy associated to an arbitrary affine Kac-Moody algebra and a choice of a vertex of its Dynkin diagram.
Compared to other proofs based on Hamiltonian formalism, our proof is entirely algebraic.
This is essentially due to the fact that the tau-structure of the DS hierarchy is an intrinsic object that encodes the
deep geometrical information of the hierarchy.

\medskip

\paragraph{\bf Organization of the paper.}
In Section \ref{sec:alg} we fix some notations and terminology and we provide the algebraic setup for the systems of partial differential equations we are interested in.
In Section~\ref{section2} we review the notion of tau-structure.
In Section~\ref{section3} we consider Miura-type transformation and prove Theorem \ref{tauint}.
In Section~\ref{section4} we revisit the example of the DS hierarchy, and 
give a new proof of its integrability as an application of Theorem \ref{tauint}.
The extension of this formalism to discrete integrable systems 
is briefly discussed in Appendix \ref{App:A}.

\medskip

\paragraph{\bf Acknowledgments.} We wish to thank Boris Dubrovin for his advice, helpful discussions and valuable suggestions.

D.V. is a member of the GNSAGA INdAM group and he acknowledges the financial support of the
project MMNLP (Mathematical Methods in Non Linear Physics) of the INFN.
D.Y. is supported in part by the CAS Project for Young Scientists in Basic Research No. YSBR-032, by 
NSFC No.~12371254, and by the National Key R and D Program of China 2020YFA0713100.

\section{Algebraic setup}\label{sec:alg}
In this section we review the algebraic setup used throughout the paper. See for example \cite{DSKV,Dickey,DZ-norm} for more details.

\subsection{Functions on space-time}\label{sec:fun}
Throughout the paper we let $\mc B$ be a given
commutative, associative, unital algebra
over a field $\mb F$
of characteristic $0$,
which consists of functions
in the space variable $ x$ and time variables $\bm t=(t_j)_{j\in E}$,
endowed with commuting derivations 
$$
\partial_x
\,\Big(=\frac{\partial}{\partial x}\Big),\,\frac{\partial}{\partial t_j}\,:\,\,\mc B\to\mc B\,,\quad j\in E\,,
$$
indexed by a countable set $E$. (Recall that, for an algebra $A$, a derivation is a linear (over $\mb F$) map $D:A\to A$ such that
$$
D(ab)=D(a)b+aD(b)
\,,
$$
for any $a,b\in A$.)
The elements of $\mc B$ are called \emph{functions on space-time}
(or simply \emph{functions}), and will be usually denoted as $f=f(x,\bm t)$.

We denote by $V$ the common kernel of all time derivatives:
$$
V
=
\bigg\{f\in\mc B\,\Big|\, \frac{\partial f}{\partial t_j}=0\,\text{ for all }\,j\in E\bigg\}
\,.
$$
We also assume that $\mc B$ is endowed with a surjective algebra homomorphism
$\mc B\twoheadrightarrow V$,
restricting to the identity map on $V\subset\mc B$,
which we shall call the \emph{evaluation at} $\bm t=0$,
and we shall denote as
$$
f=f(x,\bm t)
\,\mapsto\,
f(x,0)
=f|_{\bm t=0}
\,\in V
\,.
$$

\begin{defi}\label{def:integrable}
Given elements $f_j\in\mc B$, $j\in E$,
consider a system of equations on the unknown function $\varphi\in\mc B$:
\begin{equation}\label{20170721:eq2}
\frac{\partial\varphi}{\partial t_j}
=
f_j
\,\text{ for all }\,
j\in E
\,.
\end{equation}
The system \eqref{20170721:eq2} is called \emph{compatible}
if the following condition holds:
\begin{equation}\label{20170721:eq2b}
\frac{\partial f_i}{\partial t_j}
=
\frac{\partial f_j}{\partial t_i}
\,\text{ for all }\,
i,j\in E
\,.
\end{equation}
The algebra of functions on space-time $\mc B$ is said to be \emph{integrable}
if, for every $f(x)\in V$ and every compatible system of equations \eqref{20170721:eq2},
there exists a unique solution $\varphi\in\mc B$ such that $\varphi(x,0)=f(x)$.
\end{defi}

\begin{example}\label{20170724:ex}
The algebra $\mb F[x,t_j\mid j\in E]$ of polynomials in the variables $x$ and $t_j$, $j\in E$,
and the algebra $\mb F[[x,t_j\mid j\in E]]$ of formal power series, are both integrable.
\end{example}

\begin{example}\label{example:DYZ}
Let $V$ be an algebra of functions of $x$ closed under $\partial_x$. Then the algebra of functions
$\mc B=V[[t_j\mid j\in E]]$ is integrable, see \cite{DYZ}.
\end{example}

\subsection{Unknown variables}\label{sec:2.2}
Let~$\ell>0$ be an integer. An algebra of \emph{ordinary functions} $\mc O_{\bm u}$ in the ``unknown'' variables $\bm u=(u_1,\dots ,u_\ell)$ is an extension of the algebra of 
polynomials
$$
R_\ell=\mb F[u_\alpha\mid \alpha=1,\dots,\ell]
$$
endowed with linear maps $\frac{\partial}{\partial u_\alpha}:\mc O_{\bm u}\to\mc O_{\bm u}$, for every $\alpha=1,\dots,\ell$, which are commuting derivations
of the product in $\mc O_{\bm u}$ and extend the usual partial derivatives in $R$. 

Given an algebra of ordinary functions $\mc O_{\bm u}$ we consider the \emph{ring of differential polynomials} (with coefficients in $\mc O_{\bm u}$)
\beq\label{diffpoly}
\mathcal{A}(\mc O_{\bm u}) \= 
\mathcal{O}_{\bm u}[u_{\alpha,m}\mid \alpha=1,\dots,\ell, m\geq1]
\,.
\eeq
(If there is no ambiguity in the variables used, we will simply denote $\mc A(\mc O_{\bm u})$ by $\mc A$.)
We will often denote $u_{\alpha,0}=u_{\alpha}$, $\alpha=1,\dots,\ell$.
Note that on $\mc A$ we have natural partial derivatives $\frac{\partial}{\partial u_{\alpha,m}}:\mc A\to\mc A$ which all commute and  for every $f\in\mc A$ we have $\frac{\partial f}{\partial u_{\alpha,m}}=0$ for all but finitely many
$\alpha=1,\dots,\ell$ and $m\geq0$.
Let
\begin{equation}\label{partial}
\p \:= \sum_{\alpha=1}^\ell\sum_{m\geq 0} u_{\alpha,m+1} \frac{\p}{\p u_{\alpha,m}}\,.
\end{equation}
Clearly, $\partial$ defined in \eqref{partial} is a derivation of $\mc A$
and the pair $(\mathcal{A},\p)$ form a \emph{differential algebra}.
Note that $\partial(u_{\alpha,m})=u_{\alpha,m+1}$,
$\alpha=1,\dots,\ell$, $m\geq0$, and we have the commutation relations
$$
\left[\frac{\partial}{\partial u_{\alpha,m}},\partial\right]=\frac{\partial}{\partial u_{\alpha,m-1}}\,,
$$
where the RHS is considered to be zero if $m=0$.
Note that $\mc A$ is a particular case of an algebra of differential functions as defined in \cite{BDSK}.

In the sequel we will assume that the algebra $\mc B$ introduced in Section \ref{sec:fun} is integrable and that the following holds:
for every $\ell$-ple of functions on space-time $u_\alpha(x,\bm t)\in\mc B$, $\alpha=1,\dots,\ell$,
there exists a unique differential algebra homomorphism
$\mc A\rightarrow\mc B$,
mapping $u_{\alpha,m}\mapsto\partial_x^m u_\alpha(x,\bm t)$, $\alpha=1,\dots,\ell$, $m\geq0$,
which we shall call \emph{evaluation map},
and we shall denote as
\begin{equation}\label{20170718:eq9b}
\mc A\ni P\mapsto P|_{u_{\alpha,m}\mapsto\p_x^m(u_\alpha(x,\bt)),\, m \geq 0}\,\in\mc B
\,.
\end{equation}
By an abuse of notation we are denoting with the same symbol $u_\alpha$ an element in $\mc A$ and an element in $\mc B$. To stress the difference, in the latter case, we will always specify the dependence on the
variables $x$ and $\bm t$.
\begin{example}
Let $\mc O_{\bm u}=R_\ell$ (hence $\mc A=\mb F[u_{\alpha,m}\mid \alpha=1,\dots,\ell\,,m\geq0]$) and let $\mc B$ be any algebra of functions on space-time. Then the evaluation map \eqref{20170718:eq9b} exists for every $\ell$-ple of 
functions $u_\alpha(x,\bm t)\in\mc B$, $\alpha=1,\dots,\ell$.
\end{example}

Let us denote by $\Der(\mc A)$
the Lie algebra of all derivations of $\mc A$, and by $\Der^{\partial}(\mc A)=\{D\in\Der(\mc A)\mid [D,\partial]=0\}$
the centralizer of $\partial$ in $\Der(\mc A)$. An element $D\in\Der^{\partial}(\mc A)$ is called \emph{admissible} (it is usually called an \emph{evolutionary vector field} in the theory of integrable systems).
For $D\in\Der^{\partial}(\mc A)$ we have that
$$
D(u_{\alpha,m})=D(\partial^mu_\alpha)=\partial^m D(u_\alpha)\,,
$$
for every $\alpha=1,\dots\ell$ and $m\geq0$. Hence, by the Leibniz rule, the action of $D$ is completely determined by its values $D(u_\alpha)=W_\alpha\in\mc A$,
$\alpha=1,\dots,\ell$.
We thus have a vector space isomorphism
\begin{equation}\label{isoder}
\mc A^\ell\ni W=(W_1,\dots,W_\ell)\mapsto D_W
:=\sum_{\alpha=1}^\ell\sum_{m\geq0}\partial^m(W_\alpha)\frac{\partial}{\partial u_{\alpha,m}}\in\Der^{\partial}(\mc A)\,.
\end{equation}
In such a one-to-one correspondence, the $\ell$-ple $(u_{1,1},\dots,u_{\ell,1})$ is sent to~$\p$.

By definition, an \emph{evolution equation} on $\mc A$ has the form
\begin{equation}\label{eq:evolution}
\frac{\partial u_\alpha}{\partial t}=D(u_{\alpha})\,,\quad\alpha=1,\dots,\ell
\,,
\end{equation}
where $D$ is an admissible derivation of $\mc A$. 
Given a collection of infinitely many admissible derivations $D_j$, $j\in E$, on $\mc A$, the corresponding \emph{hierarchy of evolution equations} is the following system of evolution equations
\beq\label{pd}
\frac{\p u_\alpha}{\p t_j} \= D_j (u_\alpha)
 \,, \quad \alpha=1,\dots,\ell\,,j \in E\,.
\eeq
A \emph{solution} to the hierarchy of evolution equations \eqref{pd} is a collection of functions $u_\alpha(x,\bm t)\in\mc B$, $\alpha=1,\dots,\ell$, such that
$$
\frac{\partial u_\alpha(x,\bm t)}{\partial t_j}=D_j(u_\alpha)|_{u_{\alpha,m}\mapsto\p_x^m(u_\alpha(x,\bt)),\, m \geq 0}\,,
\quad
\alpha=1,\dots,\ell\,, j\in E\,.
$$
\begin{defi} \label{def:admissible}
A family $\{D_j\}_{j\in E}\subset\Der^{\partial}(\mc A)$ of linearly independent 
admissible derivations is called {\it integrable} if 
\begin{equation}\label{integrable}
[D_i,D_j]\=0\,,\quad \forall\,i,j\in E\,.
\end{equation}
We call the system \eqref{pd} associated to an integrable family of admissible derivations an \emph{integrable hierarchy} of evolution partial differential equations (PDEs).
\end{defi}
Note that \eqref{integrable} implies the compatibility condition \eqref{20170721:eq2b} for an integrable hierarchy.

\begin{remark}
Let $\mc B=V[[t_j\mid j\in E]]$ be as in Example \ref{example:DYZ}.
If the family $D_j$, $j\in E$, is integrable, then
for every $u(x)=(u_1(x),\dots, u_\ell(x))\in V^\ell$
there exists a solution $u(x,\bt)=(u_1(x,\bt),\dots,u_l(x,\bt))\in\mc B^\ell$ to the integrable hierarchy \eqref{pd} such that
$u(x,\bm t)|_{\bm t=0}=u(x)$.
\end{remark}

\section{Tau-structure associated to a family of derivations}\label{section2}

Let $\mc O_{\bm u}$ be an algebra of ordinary functions in the variables $u_1,\dots, u_\ell$, and let
$\mc A=\mc A(\mc O_{\bm u})$ be the differential algebra as in \eqref{diffpoly}. Consider the degree on $\mc A$ defined by
\begin{equation}\label{eq:degree}
\deg f=0\,,\quad f\in\mc O_{\bm u}\,,\quad \deg u_{\alpha,m}=m\,,\quad \alpha=1,\dots,\ell\,,m\geq1\,.
\end{equation}
We denote by~$\widehat{\mathcal{A}}\subset\mc A[[\epsilon]]$ the completion of $\mc A$ made through the gradation defined in \eqref{eq:degree}.
In other words, an element 
$a\in\widehat{\mathcal{A}}$ is an infinite series 
\begin{equation}\label{20230116:eq1}
a \= a^{[0]} \+ \epsilon a^{[1]} \+ \epsilon^2 a^{[2]}\+\cdots\,, \qquad a^{[q]}\in\mathcal{A}\,, ~ \deg a^{[q]} \= q \,.
\end{equation}
See~\cite{DZ-norm} for more details about 
$\widehat{\mc A}$.

The derivation $\partial$ in \eqref{partial} extends to a derivation $\partial:\widehat{\mathcal{A}}\rightarrow \widehat{\mathcal{A}}$ by letting $\partial(\epsilon)=0$.
Let us denote by $\Der(\widehat{\mc A})$
the Lie algebra of all derivations of $\widehat{\mc A}$ acting trivially on $\epsilon$. Clearly,
$\partial\in\Der(\widehat{\mc A})$, and we denote by
$\Der^{\partial}(\widehat{\mc A})=\{D\in\Der(\widehat{\mc A})\mid [D,\partial]=0\}$
the centralizer of $\partial$ in $\Der(\widehat{\mc A})$.
Then,
we have the vector space isomorphism
$\Der^\partial(\widehat{\mc A})\cong \widehat{\mc A}^\ell$, cf. \eqref{isoder}. Moreover, Definition \ref{def:admissible}
for an integrable family of admissible derivations extends from $\mc A$ to $\widehat{\mc A}$.

Let $D_j\in\Der^{\partial}(\widehat{\mc A})$, $j\in E$, be an integrable family of admissible derivations.
Then, we have the corresponding integrable hierarchy of evolution equations \eqref{pd} on $\widehat{\mc A}$
and we can consider the solution $u_\alpha(x,\bm t;\epsilon)\in\mc B[[\epsilon]]$, $\alpha=1,\dots,\ell$, like we did in Section~\ref{sec:2.2}.
\begin{remark}
We note that the system of PDEs \eqref{pd} corresponding to $D_j\in\Der^{\partial}(\widehat{\mc A})$ is now defined over $\widehat{\mc A}$ which is a completion of $\mc A$. The evaluation map \eqref{20170718:eq9b} extends to a differential algebra homomorphism $\widehat{\mc A}\to\mc B[[\epsilon]]$ in a natural way. When we talk about solution to the hierarchy \eqref{pd} defined over $\widehat{\mc A}$,
we mean that we can solve it on each homogeneous component, namely for each powers of $\epsilon$. Hence the solution lives in $\mc B[[\epsilon]]$.  
\end{remark}
\begin{defi}\label{def:tau}
The family $D_j\in\Der^{\partial}(\widehat{\mc A})$, $j\in E$, is called {\it admitting a tau-structure}
if there exist $\Omega_{i;j}\in \widehat{\mathcal{A}}$, $i,j\in E$,
satisfying $\partial \Omega_{i;j}\neq0$ for at least a pair $(i,j)$, $i,j\in E$, 
and such that for every $i,j,k\in E$,
\begin{align}
& \Omega_{i;j} \= \Omega_{j;i} \,, \label{tau1}\\
& D_i \, \bigl(\Omega_{j;k})
\= D_k \bigl(\Omega_{i;j}\bigr) \, . \label{tau2}
\end{align}
\end{defi}

\begin{remark}
Sometimes,  the word 
``tau-structure" is only specialized to tau-symmetry. Indeed,  
in literature, a bi-Hamiltonian tau-symmetric structure~\cite{DZ-norm} 
or a Hamiltonian tau-symmetric structure~\cite{DLYZ1} is referred to as 
a tau-structure;
the terminology tau-structure used in this paper means the 
underlying~$(\Omega_{i;j})_{i,j\in E}$.
\end{remark}
\begin{remark}
It was suggested by Boris Dubrovin that Definition \ref{def:tau} can be generalized to other rings of the jet variables $u_{\alpha,m}$ different from the ring
$\widehat{\mc A}$. In particular, he suggested to investigate the tau-structure for hierarchies of
Krichever--Novikov type~\cite{KN}. To this aim one should generalize the previous construction in the more general context
of algebras of differential functions \cite{BDSK}.
\end{remark}
Let $D_j\in\Der^{\partial}(\widehat{\mc A})$, $j\in E$, be an integrable family of admissible derivations admitting a tau-structure $\Omega_{i,j}\in\widehat{\mc A}$, $i,j\in E$. By the integrability assumption on $\mc B$ and equations \eqref{tau1}-\eqref{tau2}, for an arbitrary solution $u(x,\bt;\epsilon)\in\mc B[[\epsilon]]^\ell$ to the integrable hierarchy~\eqref{pd} there exists $\tau=\tau(x,\bt;\epsilon)$ (possibly in some algebra extension of $\mc B((\epsilon))$) satisfying 
\beq\label{eq:tau}
\epsilon^2
\frac{\p^2\log \tau(x,\bt;\epsilon)}{\p t_i \p t_j} 
\= \Omega_{i;j}|_{u_{\alpha,m}\mapsto\p_x^m(u_\alpha(x,\bt;\epsilon)),\, m \geq 0}\,.
\qquad
i,j\in E\,,
\eeq
We call $\Omega_{i;j}|_{u_{\alpha,m}\mapsto\p_x^m(u_\alpha(x,\bt;\epsilon))}\in\mc B[[\epsilon]]$, $i,j\in E$, the 
\emph{two-point correlation functions} of the solution~$u(x,\bt;\epsilon)$. We call $\tau(x,\bt;\epsilon)$ the \emph{tau-function} 
of the solution~$u(x,\bt;\epsilon)$, although it is uniquely determined by~$u(x,\bt;\epsilon)$ up to 
multiplying by a factor of the form
\[ \exp \Bigl(b(x;\epsilon)+\sum_{j\in E} a_j(x;\epsilon) t_j\Bigr) \,,\]
where $b(x;\epsilon), a_j(x;\epsilon)\in V[[\epsilon]]$. See \cite{DYZ} for further details.

\section{Miura-type transformations and tau-coordinates}\label{section3}

In this section we let $\mc O_{\bm u}$ be an algebra of ordinary functions in the variables $u_1,\dots, u_\ell$, and we denote by
$\mc A_{\bm u}=\mc A(\mc O_{\bm u})$ the differential algebra in \eqref{diffpoly} (we emphasize the dependence of $\bm u$ in the notation since later we will use other independent variables),
and by $\widehat{\mc A}_{\bm u}$ its completion defined in Section \ref{section2}.
We will also use the notation
$$
\bm u_m=(u_{1,m},\dots, u_{\ell,m})\,,
\qquad
m\geq0
\,.
$$
Let us start with recalling the following definition from \cite{DZ-norm}.
\begin{defi}
An $\ell$-ple 
\[\bigl(V_1(\bm u,\bm u_{1},\bm u_{2},\dots;\epsilon),\dots,V_\ell(\bm u,\bm u_{1},\bm u_{2},\dots;\epsilon)\bigr)
\in \widehat{\mathcal{A}}_{\bm u}^\ell
\] 
is called {\it of Miura-type} if 
\beq\label{Jacobiannondegen}
\det \, \biggl(\frac{\p V_{\alpha}^{[0]}}{\p u_\beta}\biggr) \; \neq \; 0 \,.
\eeq
(In \eqref{Jacobiannondegen} we are expanding $V_{\alpha}$ as in \eqref{20230116:eq1}.)
\end{defi}
Let $\widetilde{\mc O}_{\bm v}$ be an algebra of ordinary functions in the variables $\bm v=(v_1,\dots, v_\ell)$,
and let  $\bigl(\mathcal{J}_{\bm v}=\mc A(\widetilde{\mc O}_{\bm v}),\tilde \partial\bigr)$ 
be the differential algebra in \eqref{diffpoly},
where $\tilde \partial=\sum_{\alpha=1}^\ell\sum_{m\geq 0} v_{\alpha,m+1} \frac{\p}{\p v_{\alpha,m}}$.
Let $\bm V:=(V_\alpha(\bm u,\bm u_{1},\bm u_{2},\dots;\epsilon))_{\alpha=1,\dots,\ell}\in\widehat{\mathcal{A}}_{\bm u}^{\ell}$
be a Miura-type $\ell$-ple. We associate with 
this Miura-type $\ell$-ple a differential algebra homomorphism $\phi_{\bm V}: \widehat{\mathcal{J}}_{\bm v}
\rightarrow\widehat{\mathcal{A}}_{\bm u}$, defined
by assigning  on generators
\beq\label{miuravu}
v_{\alpha}\in\widehat{\mathcal{J}}_{\bm v}  \quad \mapsto \quad
\phi_{\bm V}(v_{\alpha})=V_{\alpha}(\bm u,\bm u_{1},\bm u_{2},\dots;\epsilon) \in\widehat{\mathcal{A}}_{\bm u} \,,
\eeq
and by extending it using the Leibniz rule. This homomorphism is called a \emph{Miura-type map}.
Recall from \cite{DZ-norm} that for every $\alpha\in\{1,\dots,\ell\}$ there exists a unique element 
$U_\alpha(\bm v,\bm v_{1},\bm v_{2},\dots;\epsilon)\in\widehat{\mathcal{J}}_{\bm v}$ such that 
\beq\label{Uu}
\phi_{\bm V}(U_\alpha(\bm v,\bm v_{1},\bm v_{2},\dots;\epsilon))
\=u_\alpha \,.
\eeq
The element  $\bm U=(U_\alpha(\bm v,\bm v_1,\bm v_2,\dots;\epsilon))_{\alpha=1,\dots,\ell}\in\widehat{\mathcal{J}}_{\bm v}^{\ell}$
gives a differential algebra homomorphism $\psi_{\bm U}: \widehat{\mathcal{A}}_{\bm u}\rightarrow \widehat{\mathcal{J}}_{\bm v}$ by assigning on generators
\[
u_{\alpha}\in\widehat{\mathcal{A}}_{\bm u}  \quad \mapsto \quad
\psi_{\bm U}(u_{\alpha})=U_{\alpha}(\bm v,\bm v_{1},\bm v_{2},\dots;\epsilon)\in\widehat{\mathcal{J}}_{\bm v} \,,
\]
and by extending it using the Leibniz rule. 
One can verify using \eqref{miuravu} and \eqref{Uu} that 
$\phi_{\bm V}\circ \psi_{\bm U}={\rm id}_{\widehat{\mc A}_{\bm u}}$ and
$\partial=\phi_{\bm V}\circ\tilde\partial\circ \psi_{\bm U}$.
Similarly, one can show that $\psi_{\bm U}\circ \phi_{\bm V} = {\rm id}_{\widehat{\mc J}_{\bm v}}$
and $\tilde \partial=\psi_{\bm U}\circ\partial\circ \phi_{\bm V}$.
 We call $\psi_{\bm U}$ the inverse Miura-type map of~$\phi_{\bm V}$. For more details of Miura-type and inverse Miura-type maps see~\cite{DZ-norm}.

Let~$D$ be an admissible derivation on~$\widehat{\mathcal{A}}_{\bm u}$. It induces 
an admissible derivation on~$\widehat{\mathcal{J}}_{\bm v}$, denoted by~$\widetilde{D}$, 
via specifying 
\beq\label{Dtilde}
\widetilde{D}(v_\alpha) \:= \psi_{\bm U}\circ D\circ\phi_{\bm V}(v_\alpha)
\,,
\quad \alpha=1,\dots,\ell\,,
\eeq
and extending by the Leibinz rule.
\begin{theorem}\label{tauint}
Let $E$ be an infinite set with a distinguished element~${\bm1}\in E$,  
and let $D_j$, $j\in E$, be a family of linearly independent admissible 
derivations on~$\widehat{\mathcal{A}}_{\bm u}$, 
such that $D_{\bm1}$ commutes with all admissible derivations on~$\widehat{\mathcal{A}}_{\bm u}$. 
Suppose that the family $D_j$, $j\in E$, admits a tau-structure~$\Omega_{i;j}\in\widehat{\mathcal A}_{\bm u}$, $i,j\in E$,
such that
\beq\label{mondegen}
\bm V:= (\Omega_{\bm 1;i_1},\dots,\Omega_{\bm 1;i_\ell}) \mbox{ is a 
Miura-type $\ell$-ple},
\eeq
for some distinct $i_1,\dots,i_\ell\in E$.
Then the following statements are true:
\begin{enumerate}[(i)]
\item the family $D_j$, $j\in E$, is integrable;
\item let $\phi_{\bm V}: \widehat{\mathcal{J}}_{\bm v}\rightarrow\widehat{\mathcal{A}}_{\bm u}$ be the Miura-type map 
generated by~\eqref{miuravu},
and let~$\bm U=(U_\alpha)_{\alpha=1,\dots,\ell}$, be 
the unique element in~$\widehat{\mathcal{A}}_{\bm u}^\ell$ satisfying~\eqref{Uu}.
Then we have
$$
D_j(u_\alpha)  \= \sum_{\substack{\beta=1,\dots,\ell\\m\geq 0}}
\phi_{\bm V}\biggl(\frac{\p U_\alpha}{\p v_{\beta,m} } \biggr)
\partial^m(D_{\bm1}(\Omega_{j;i_\beta})) \,,
\quad j\in E\,,\alpha=1,\dots,\ell\,.
$$
\end{enumerate}
\end{theorem} 
\begin{proof}
For every $i,j\in E$ and $\alpha=1,\dots,\ell$, using \eqref{tau1}, \eqref{tau2} and the assumption that $D_{\bm 1}$ commutes
with all admissible derivations, we have
$$
D_i D_j \bigl(\Omega_{\bm1,i_\alpha}\bigr) \= D_i D_{\bm1} \bigl(\Omega_{j;i_\alpha}\bigr)
\= D_{\bm 1} D_i \bigl(\Omega_{j;i_\alpha}\bigr) \= D_{\bm 1} D_{i_\alpha}  \bigl(\Omega_{i;j}\bigr)
\,,
$$ 
which is symmetric in exchanging $i$ and $j$, namely
\begin{equation}\label{20230306:eq1}
D_i D_j \bigl(\Omega_{\bm1,i_\alpha}\bigr)
 \= D_j D_i \bigl(\Omega_{\bm 1;i_\alpha}\bigr)
\,.
\end{equation}
Denote by~$\widetilde{D}_j$ the derivations on~$\widehat{\mathcal{J}}_{\bm v}$ defined using \eqref{Dtilde}. For every
$\alpha=1,\dots,\ell$, we have
$$
\widetilde{D}_j(v_\alpha)=\psi_{\bm U}\circ D_j(\Omega_{\bm 1,j_\alpha})
\,.
$$
Using the fact that
$\psi_{\bm U}$ and $\phi_{\bm V}$ are inverse maps to each other and equation \eqref{20230306:eq1} we get
$$
[\widetilde{D}_i,\widetilde{D}_j](v_\alpha)=\psi_{\bm U}\circ [D_i,D_j](\Omega_{\bm 1,j_\alpha})=0\,,
\quad
\alpha=1,\dots,\ell\,.
$$
Since $\bigl[\widetilde{D}_i, \widetilde{D}_j\bigr]$ is a derivation commuting with $\tilde \partial$, we conclude that $\bigl[\widetilde{D}_i, \widetilde{D}_j\bigr]=0$. Therefore also $[ D_i, D_j ]=0$. 
Finally, we have
\begin{align}
 D_j (u_\alpha)&= \phi_{\bm V}\circ \widetilde D_j \circ \psi_{\bm U} (u_\alpha)
 =  \phi_{\bm V}\circ \widetilde D_j (U_\alpha)
 =\phi_{\bm V}\biggl(\sum_{\substack{\beta=1,\dots,\ell\\m\geq 0}}  \frac{\p U_\alpha}{\p v_{\beta,m} }  \widetilde D_j (v_{\beta,m}) \biggr) \nn\\
& =  \sum_{\substack{\beta=1,\dots,\ell\\m\geq 0}}  \phi_{\bm V} \biggl(\frac{\p U_\alpha}{\p v_{\beta,m} }\biggr) \, \phi_{\bm V} \bigl( \tilde{\partial}^m \left( D_j (v_{\beta})\right)\bigr)  
\nn\\
&=  \sum_{m\geq 0}  \phi_{\bm V} \biggl(\frac{\p U_\alpha}{\p v_{\beta,m} }\biggr) \, 
(\phi_{\bm V} \circ \tilde{\partial}^m \circ \psi_{\bm U})  \left(D_j (\Omega_{\bm 1;i_\beta}) \right)\nn\\
& =  \sum_{\substack{\beta=1,\dots,\ell\\m\geq 0}}  \phi_{\bm V} \biggl(\frac{\p U_\alpha}{\p v_{\beta,m} }\biggr) \, 
\partial^m(D_{\bm 1}(\Omega_{j;i_\beta})) \,.\nn
\end{align} 
The theorem is proved.
\end{proof}
\noindent We call~$\bm v=(v_1,\dots,v_{\ell})$ 
in the above theorem a choice of {\it tau-coordinates}.

\begin{remark}
Statement~(i) of Theorem~\ref{tauint} provides a way of proving integrability of the hierarchy of evolution equations
\eqref{pd} corresponding to the family $D_j$, $j\in E$. This 
principle has been used successfully in~\cite{DYZ}.
Statement (ii) provides a way of getting explicit expressions for the integrable family of admissible derivations $D_j$, $j\in E$, starting from the tau-structure $\Omega_{i,j}$,
$i,j\in E$.
\end{remark}

\section{Example: Drinfeld--Sokolov hierarchy}\label{section4}
In this Section we give an application of Theorem \ref{tauint}. We review the construction of the tau-structure of the DS hierarchy in Section \ref{sec:5.2}
and we use it to construct a choice of tau-coordinates. 
Using the material in Section \ref{sec:5.4} we finally prove the integrability of DS hierarchies in Section \ref{sec:int}.

\subsection{Preliminaries}
The material of this subsection  can be found in textbooks on Kac-Moody algebras, see for example~\cite{Kac94}.

\subsubsection{Principal and standard gradations of an affine Kac--Moody algebra}

Let $X_n^{(r)}$ be an affine Kac-Moody algebra of rank~$\ell$ 
with $r=1,2,3$, with Chevalley generators 
$\{e_i\,,h_i\,,f_i\mid i=0,\dots,\ell\}$, and let $a_i$ (respectively $a_i^\vee$) be the Kac labels 
for the associated Dynkin diagram.
Denote by $h, h^\vee$ the 
 \emph{Coxeter number} and the \emph{dual Coxeter number} of $X_{n}^{(r)}$.

Let $\widetilde{\mf g}$ be the quotient of $X_{n}^{(r)}$ by the one dimensional space generated by the
central element
$K=\sum_{i=0}^\ell a_i^\vee h_i$.
The \emph{principal gradation} on $\widetilde{\mf g}$ is defined by
assigning
$$\deg^{\rm pr} e_i=-\deg^{\rm pr} f_i=1\,,\quad i=0,\dots,\ell\,.$$
Let $\widetilde{\mf g}^{\,k}$ be
the subspace of homogeneous elements of principal degree $k$.
We will work with the following completion of $\widetilde{\mf g}$, which we denote with the same symbol by an abuse of notation,
\begin{equation}\label{deg:principal}
\widetilde{\mf g}\=\wbigoplus_{k\in\mb Z}\,\widetilde{\mf g}^{k}
\,,
\end{equation}
where the direct sum is completed by allowing infinite series in negative degree.
Given an element $a\in\widetilde{\mf g}$ we denote by $a^+$ its projection on $\widetilde{\mf g}^{\geq0}=\oplus_{k\geq0}\widetilde{\mf g}^k$ and 
by $a^{-}$ its projection on $\widetilde{\mf g}^{<0}=\widehat{\oplus}_{k<0}\widetilde{\mf g}^k$.

Let $m$ be an integer from $0$ to $\ell$. 
Take $c_m$ to be the $m$th vertex of the Dynkin diagram of~$X_{n}^{(r)}$.
The \emph{standard gradation} corresponding to
$c_m$ is defined by assigning
$$\deg_{\rm st} e_i=-\deg_{\rm st} f_i=\delta_{i,m}\,,\quad i=0,\dots,\ell\,.$$
Then, we also have the direct sum decomposition
\begin{equation}\label{deg:standard}
\widetilde{\mf g}\=\wbigoplus_{k\in\mb Z}\,\widetilde{\mf g}_{k}
\,,
\end{equation}
where $\widetilde{\mf g}_k$ denotes the homogeneous subspace of elements with standard degree $k$.
Given an element $a\in\widetilde{\mf g}$ we denote by $a_+$ its projection on $\widetilde{\mf g}_{\geq0}=\oplus_{k\geq0}\widetilde{\mf g}_k$ and 
by $a_{-}$ its projection on $\widetilde{\mf g}_{<0}=\widehat{\oplus}_{k<0}\widetilde{\mf g}_k$.

\subsubsection{Loop realization of~$X_n^{(r)}$}
Let $\mf g$ be a simple finite dimensional Lie algebra of rank~$n$, and let $\sigma$ be an automorphism of
$\mf g$ satisfying $\sigma^N=1$
for a positive integer~$N$. 
Since $\sigma$ is diagonalizable, we have the direct sum decomposition
$$
\mf g\=\bigoplus_{\bar k\in\mb Z/N\mb Z}\mf g_{\bar k}
\,,
$$
where $\mf g_{\bar{k}}$ is the eigenspace of $\sigma$ with eigenvalue $e^{\frac{2\pi i k}{N}}$.

Denote by $L(\mf g)=\mf g\otimes\mb C((\lambda^{-1}))$ the space of Laurent series in the variable $\lambda^{-1}$
with coefficients in~$\mf g$.
The Lie algebra structure of $\mf g$ extends naturally to a Lie algebra structure on~$L(\mf g)$.
We extend $\sigma$ to a Lie algebra homomorphism $\sigma:L(\mf{g})\to L(\mf g)$ as follows
$$
\sigma(a\otimes f(\lambda))\=\sigma(a)\otimes f(e^{-\frac{2\pi i}{N}}\lambda),
$$
for $a\in \mf g$, $f\in \mb C((\lambda^{-1}))$.
The subalgebra of invariant elements with respect to $\sigma$ is the twisted algebra of Laurent series in the variable $\lambda^{-1}$
with coefficients in $\mf g$, and we denote it by
$$
L(\mf g,\sigma)\=\left\{a\in L(\mf{g})\mid \sigma(a)=a\right\}
\,.
$$
For every vertex $c_m$ of the Dynkin diagram of $X_n^{(r)}$ there exists a simple Lie algebra~$\g$ 
(independent of~$m$) and an automorphism $\sigma_m$
of $\mf g$
of order $N_m=ra_m$, such that $\widetilde{\mf g}\cong L(\mf g,\sigma_m)$, 
and the gradation of $L(\mf g,\sigma_m)$ in powers of~$\lambda$ 
is the standard gradation of~$\tilde\g$. 
We call $L(\mf g,\sigma_m)$ the \emph{standard realization} of $\widetilde{\mf g}$ corresponding to the vertex $c_m$,
and we denote
\begin{equation}\label{dec:sigmatwisted}
L(\mf g,\sigma_m)
\=\wbigoplus_{k\in \mb Z}\, L(\mf g,\sigma_m)_k
\end{equation}
with $L(\mf g,\sigma_m)_k:=\mf g_{\bar k}\otimes \lambda^k$ (cf.~equation~\eqref{deg:standard}).

Recall that $\mf a:=L(\mf g,\sigma_m)_0=\mf g_{\bar0}$ is a semisimple Lie algebra. 
Let us denote by 
\begin{equation}\label{eq:e2022}
e\=\sum_{i\neq m}^\ell e_i\in \mf a \subset \widetilde{\mf g}
\end{equation}
a principal nilpotent element of~$\mf a$. 
By the Jacobson--Morozov Theorem there exist a nilpotent element~$f$ and 
a semisimple element~$\rho^\vee$ in $\mf a$ such that
$$
[\rho^\vee,e]\=e\,,
\qquad
[\rho^\vee,f]\=-f
\,,\qquad
[e,f]\=\rho^\vee
\,.
$$
Note that $\rho^\vee$ is also a semisimple element of~$\g$, and 
 the principal gradation \eqref{deg:principal} on $\widetilde{\mf g} \cong L(\g, \sigma_m)$ is defined by the 
following linear map
$$
\ad\rho^\vee+\frac{rh}{N_m}\lambda\frac{d}{d\lambda}:
L(\g, \sigma_m) \to L(\g, \sigma_m)
\,.
$$
As in~\eqref{deg:principal} we write
$$
L(\mf g,\sigma_m)=\wbigoplus_{k\in\mb Z}\, L(\mf g,\sigma_m)^k
\,,
$$
where elements in $\widetilde{\mf g}^k\cong L(\mf g,\sigma)^k$ have principal degree~$k$.

Let $(\cdot\,|\,\cdot)$ be the normalized invariant bilinear form on $\mf g$.
We extend it to a bilinear form on $L(\mf g)$ with values in $\mb C((\lambda^{-1}))$ by
$$
(a\otimes f(\lambda)|b\otimes g(\lambda))\=(a|b)f(\lambda)g(\lambda)\,,
\qquad
a,b\in\mf g\,,
f(\lambda),g(\lambda)\in\mb C((\lambda^{-1}))
\,.
$$
In the sequel we will consider the restriction of this $\mb C((\lambda^{-1}))$-valued bilinear form to $\widetilde{\mf g}\subset L(\g)$.

As in~\cite{DVY} let us introduce the linear map 
$\pi_\lambda:\mb C((\lambda^{-1}))\to \lambda^{-1}\mb C[[\lambda^{-N_m}]]$ defined via 
$$
\lambda^k\mapsto
\left\{\begin{array}{ll}
\lambda^k\,, & \text{if }k\equiv -1 \, (\bmod N_m), ~ k<0\,,
\\
0\,, & \text{otherwise}
\,,
\end{array}
\right.
$$
where  $k\in\mb Z$. Similarly, we define
$\pi_{\lambda,\mu}=\pi_\lambda\circ\pi_\mu$
and $\pi_{\lambda,\mu,\eta}=\pi_\lambda\circ\pi_\mu\circ\pi_\eta$
(see \cite{DVY} for details).

\subsubsection{Principal Heisenberg subalgebra}
Introduce the {\it cyclic element}
\begin{equation}\label{eq:Lambda}
\Lambda(\lambda) \= \sum_{i=0}^\ell e_i \= e \+ e_m(\lambda) \in\widetilde{\mf g} \,,
\end{equation}
where $e$ is defined in~\eqref{eq:e2022}. By~\eqref{dec:sigmatwisted}
we have that $e_m(\lambda)=\tilde e_m\lambda$, for some $\tilde e_m\in\mf g_{\bar1}$.
Note that $\Lambda$ is homogeneous of principal degree~$1$.
Let $\mc H=\Ker \, \ad \, \Lambda(\lambda)$ be the so-called \emph{principal (centerless) Heisenberg subalgebra} of $\widetilde{\mf g}$.
Recall that, $\mc H$ is abelian
and we have the direct sum decomposition
\begin{equation}\label{dec:Lambda}
\widetilde{\mf g}\=\mc H \,\oplus \, \im\,\ad\,\Lambda(\lambda)
\,.
\end{equation}
Given $A\in \widetilde{\mf g}$ we denote by $\pi_{\mc H}(A)\in\mc H$ its projection with respect to the direct sum
decomposition \eqref{dec:Lambda}. 

It is known that 
$\mc H$ admits the following decomposition:
$$
\mc H\= 
\wbigoplus_{i\in E}\mb C \Lambda_i(\lambda)\,,
$$
where 
$E=\bigcup_{a=1}^n(m_a+rh\mb Z)$
is the set of exponents of $X_{n}^{(r)}$,
and $\Lambda_i(\lambda)\neq 0$, $i\in E$, have principal degree~$i$. 
Here,
we have that
$1=m_1<m_2\leq\dots\leq m_{n-1}<m_n=rh-1$
satisfy
$
m_{a}+m_{n+1-a}=rh
$,
$a=1,\dots,n,$
and we normalize the elements $\Lambda_i(\lambda)$, $i\in E$, as follows:
\begin{align}
& \Lambda_1(\lambda)\=\Lambda(\lambda)\notag\,, \\
& (\Lambda_{m_a}(\lambda)|\Lambda_{m_b}(\lambda))\=\delta_{a+b,n+1} h\lambda^{N_m}\,, \quad 1\leq a,b\leq n\,, \notag\\
& \Lambda_{m_a+rhk}(\lambda)\=\Lambda_{m_a}(\lambda)\lambda^{kN_m}\,, \quad k\in\mb Z\,,1\leq a,b\leq n\,. \label{eq:normalization1}
\end{align}

\subsection{Pre-DS derivations, basic resolvents and tau-structure}\label{sec:5.2}
Recall that $\mf a=\g_{\bar{0}}$ is a semi-simple Lie algebra.
With respect to the principal gradation we can write
\begin{equation}\label{dec:dynkin}
\mf a\=\bigoplus_{i=-h_{\mf a}+1}^{h_{\mf a}-1}\mf a^i
\,,
\end{equation}
where $\mf a^i=\mf a\cap\widetilde{\mf g}^{\,i}$ and $h_{\mf a}$ is the Coxeter number of $\mf a$.
Let us further denote by $\mf n$ the nilpotent subalgebra $\mf n=\mf a^{<0}$ of~$\mf a$ and by~$\mf b$ the 
Borel subalgebra $\mf b=\mf n\oplus\mf a^0$ of~$\mf a$. 
We note that we have the following commutation relations 
\begin{equation}\label{20171109:eq1}
[\mf n,\mf b]\subset\mf n
\,,
\qquad
[\mf n,e_m(\lambda)]\=0
\,,
\qquad
[\mf n,e_i]\subset\mf b
\,,
\qquad
i\in\{0,\dots,\ell\}\setminus\{m\}
\,,
\end{equation}
and that
$\mf a^0$ is the Cartan subalgebra of $\mf a$ generated by the Chevalley generators $h_i$, $0\leq i\leq \ell$, $i\neq m$. 

Let $v_1,\dots,v_{\dim \fb}$ be a basis of~$\fb$, and $p_1,\dots,p_{\dim \fn}$ be a basis of~$\fn$, both homogeneous with respect to 
the gradation \eqref{dec:dynkin}.
Following~\cite{DS}, let  
\begin{equation}\label{eq:lax}
\mathcal{L} \:= \p \+ \Lambda(\lambda) \+ q
\end{equation}
be the Lax operator of the pre-DS hierarchies.
Here, $q:=\sum_{i=1}^{\dim \fb} q_i v_i$, where $q_1,\dots,q_{\dim \fb}$ are indeterminates.

Let
$\mathcal{A}_{\bm q}=\mathcal{A}(\mathcal{O}_{\bm q})$ be a differential algebra in the variables
$\bm q=(q_1,\dots,q_{\dim \fb})$ introduced in \eqref{diffpoly}.

\begin{defi}[\cite{BDY, DVY}] \label{resol-defi}
The {\it basic resolvents}~$R_{m_a}(\lambda)\in  \mathcal{A}_{\bm q} \otimes  \g((\lambda^{-1})) $, $a=1,\dots,n$, associated with~$\L$ are defined as the unique elements satisfying:
\begin{align}
& [\L,R_{m_a}(\lambda)] \= 0\, , \label{basicdef1}\\
& R_{m_a}(\lambda) \= \Lambda_{m_a}(\lambda) \+ \mbox{ lower order terms w.r.t. } \deg^{\rm pr} \, , \label{basicdef2} \\
& \bigl(R_{m_a}(\lambda)\, | \,R_{m_b}(\lambda)\bigr) \= h \, \delta_{a+b,n+1} \, \lambda^{N_m} \, . \label{boundaryr}
\end{align}
\end{defi}
Note that Definition \ref{resol-defi} is stated differently, but equivalently, in \cite{DVY}. In fact, it is proven in \cite{DS}, that
the space of elements in $\mathcal{A}_{\bm q} \otimes  \g((\lambda^{-1})) $ commuting with $\mc L$ is
$
e^{-\ad U(\lambda)}(\mathcal H)\,,
$
where $U(\lambda)\in\mathcal{A}_{\bm q} \otimes (\im\ad\Lambda(\lambda))^{<0}$ is the unique element such that
\begin{equation}\label{eq:fundamental}
e^{\ad U(\lambda)}(\partial+q+\Lambda(\lambda))=\partial+\Lambda(\lambda)+H(\lambda)\,,
\end{equation}
for some $H(\lambda)\in\mathcal{A}_{\bm q}\otimes\mathcal H^{<0}$. Hence, by the properties of the exponential map
and the normalization \eqref{eq:normalization1}, $R_{m_a}(\lambda)=e^{-\ad U(\lambda)}(\Lambda_{m_a}(\lambda))$, $a=1,\dots,n$, is the unique solution
to \eqref{basicdef1}-\eqref{boundaryr}. 

Using the fact that $[\lambda^{kN_m} R_{m_a}(\lambda)\,,\, \L]=0$, $a=1,\dots, n$, $k\in\mb Z$,
it is further proven in \cite{DS} that 

\begin{equation}\label{DRL}
\bigl[\bigl(\lambda^{kN_m} R_{m_a}(\lambda)\bigr)_+ \,,\, \L\bigr] \in \mathcal{A}_q\otimes \fb
\,.
\end{equation}
Here,  $(\,)_+$ refers to taking the non-negative part with respect to the standard gradation \eqref{dec:sigmatwisted},
that is the polynomial part in~$\lambda$; similarly, 
$(\,)_-$ refers to taking the negative part with respect to the standard gradation, that is the terms having negative powers in~$\lambda$.

By using equation \eqref{DRL}, we can 
define a family of admissible derivations~$D_{a,k}^{\rm pre}$ on $\mathcal{A}_{\bm q}$, for $a=1,\dots,n$, $k\geq0$, via requiring
\beq\label{pre-DS}
 D_{a,k}^{\rm pre} (q) \=
\Bigl[\bigl(\lambda^{kN_m} R_{m_a}(\lambda)\bigr)_+ \,,\, \L\Bigr] \,.
\eeq
\begin{defi} [\cite{BDY, DVY}]\label{tau-symm-def}
Define a family of elements 
$\Omega_{a,k_1;b,k_2}\in{\mathcal A}_{\bm q}$, $a,b=1,\dots,n$, $k_1, k_2\geq 0$, by
\beq\label{tau-symm-eq}
\sum_{k_1,k_2\geq 0} \frac{\Omega_{a,k_1;b,k_2}}{\lambda^{k_1 N_m+1} \mu^{k_2 N_m+1}} 
\=  \pi_{\lambda,\mu} \Biggl(\frac{\bigl(R_{m_a} (\lambda) \, | \, R_{m_b} (\mu)\bigr)}{(\lambda-\mu)^2} \,-\,
\frac{\delta_{a+b,n+1}}r \frac{m_a  \lambda^{N_m} + m_b \mu^{N_m}}{(\lambda-\mu)^2}\Biggr).
\eeq
\end{defi}

It follows immediately from the defining equation~\eqref{tau-symm-eq} and the symmetry property of the bilinear form~$(\cdot|\cdot)$ that 
\beq\label{symm1}
\Omega_{a,k_1; b,k_2}\=\Omega_{b,k_2; a,k_1}\,, \quad \text{for all } a,b=1,\dots,n\text{ and }k_1,k_2\geq0\,.
\eeq

To proceed let us introduce the operators:
$$
\nabla_a(\lambda) \= 
\sum_{k\geq 0}  \frac{D_{a,k}^{\rm pre}}{\lambda^{N_m k+1}} \,, \quad a=1,\dots,n\,.
$$
It was proven in~\cite[Lemma~5]{DVY} that for $a,b=1\dots,n$, we have
\begin{equation}\label{20171104}
\nabla_a(\lambda)R_{m_b}(\mu)\=\pi_{\lambda}\frac{[R_{m_a}(\lambda),R_{m_b}(\mu)]}{\lambda-\mu}\,.
\end{equation}
Using~\eqref{tau-symm-eq} and~\eqref{20171104} 
it is also proven in \cite[Proposition~1]{DVY}
that
$$
\nabla_c(\xi)\, \Biggl(\sum_{k_1,k_2\geq 0}  
\frac{\Omega_{a,k_1;b,k_2}}{\lambda^{N_m k_1+1} \mu^{N_m k_2+1}} \Biggr) \= 
- \pi_{\lambda,\mu,\xi}\frac{\bigl([R_{m_c}(\xi),R_{m_a}(\lambda)] \, | \, R_{m_b}(\mu)\bigr)} 
{(\lambda-\mu)(\mu-\xi)(\xi-\lambda)} \,,
$$
for all $a,b,c=1,\dots,n$.
It follows immediately that 
\begin{align}
& D_{a,k_1} (\Omega_{b,k_2;c,k_3})\=  D_{c,k_3} (\Omega_{a,k_1;b,k_2})  \,,
\label{symm2}
\end{align}
for all $a,b,c=1,\dots,n$ and $k_1,k_2,k_3\geq0$. 

\subsection{Gauge invariants}\label{subsectiongaugeinv}
Let $(p_i)_{i=1,\dots,\dim \fn}$ be the basis of $\fn$ introduced in Section~\ref{sec:5.2}, and let~$S= \sum_{i=1}^{\dim \fn} S_i p_i$, where    
$S_1,\dots,S_{\dim\mf n}$ are indeterminates. Denote by 
$\mathcal{A}_{\bm q,\bm S}=\mc A(\mc O_{\bm q,\bm S})$ a differential algebra in the variables $\bm q=(q_1,\dots, q_{\dim\mf b})$ and
$\bm S=(S_1,\dots,S_{\dim \mf n})$ introduced in \eqref{diffpoly}, extending $\mc A_{\bm q}$.
 Define $Q=\sum_{i=1}^{\dim \fb} Q_i v_i\in\mathcal{A}_{\bm q,\bm S} \otimes \fb$ by 
\beq\label{eq:gauge}
e^{\ad S} {\mathcal L}
\= e^{\ad S}(\partial + \Lambda(\lambda)+q )\= \p \+ \Lambda(\lambda) \+ Q\,.
\eeq
Using the element~$Q$ we introduce a differential algebra homomorphism 
\[f:\mathcal{A}_{\bm q}\to \mathcal{A}_{\bm q,\bm S}\,,\] defined by assigning on generators
\[
q_{i} \; \mapsto \; Q_i \,, \quad i=1,\dots,\dim_{\fb}\,,
\]
and by extending it using the Leibniz rule.
The following definition is due to Drinfeld and Sokolov.
\begin{defi}[\cite{DS}]
An element $w \in \mathcal{A}_{\bm q}$ is called a {\it gauge invariant} if 
$$
f(w)  \=  w \,.
$$
\end{defi}
Denote by~$\mathcal{R}=\{w\in\mc{A}_{\bm q}\mid f(w)=w\}\subset\mathcal{A}_{\bm q}$ the set of all gauge invariants. 
Clearly, $\mc R$ is a differential subalgebra of $\mc A_{\bm q}$.
Following Drinfeld and Sokolov~\cite{DS}, we show that $\mathcal{R}$ is a freely 
{\it generated} differential algebra through an explicit construction of the generators.
Recall that 
a linear subspace $V\subset\fb$ is called a \textit{gauge of DS-type} if 
\beq\label{DSgaugeequal}
\fb \= V \oplus [e, \fn]\,.
\eeq
From now on we fix~$V\subset\fb$ to be a gauge of DS-type.
\begin{prop}[\cite{DS}]\label{NQ}
There exists a unique pair $(S_{\rm can},Q_{\rm can})$ with $S_{\rm can}\in\mc A_{\bm q,\bm S}\otimes\mf n$
and $Q_{\rm can}\in\mc A_{\bm q,\bm S}\otimes V$, such that 
\beq\label{canform}
e^{\ad {S_{\rm can}}} {\mathcal L} \= \p \+ \Lambda(\lambda) \+ Q_{\rm can} \,.
\eeq
Moreover, $S_{\rm can}\in\mc A_{\bm q}\otimes\mf n$
and $Q_{\rm can}\in\mc R\otimes V$.
\end{prop}
\begin{proof}
We give the detail of the proof for completeness.
Recalling that $\Lambda= e + e_m(\lambda)$ and $[\mf n,e_m(\lambda)]=0$, we find that 
equation~\eqref{canform} is equivalent to
\[
e^{\ad{S_{\rm can}}} (\p + e +q) \= \p \+ e \+ Q_{\rm can} \,.
\]
By a straightforward computation the above equation is further equivalent to
\beq\label{20230403:eq1}
e\+ Q_{\rm can} \= \sum_{m=0}^{\infty} 
 \frac{(\ad S_{\rm can})^m (\p(S_{\rm can}))}{(m+1)!} 
 \+ 
e^{\ad{S_{\rm can}}}  (e + q) \,.
\eeq
For $A\in\mc A_{\bm q,\bm s}\otimes \mf g$, let us denote by $A^{\{k\}}\in\mc A_{\bm q,\bm s}\otimes \mf g^k$, $k\in\mb Z$, the homogeneous component  of $A$ with principal degree $k$.
Equation \eqref{20230403:eq1} can be solved recursively by comparing the terms with principal degree~$-k$, $k\geq0$, on both sides.
For $k=0$, we have
\beq\label{Qq00}
Q_{\rm can}^{\{0\}} \+ \Bigl[e,S_{\rm can}^{\{-1\}}\Bigr] \= q^{\{0\}}\,.
\eeq
From~\eqref{DSgaugeequal} 
we then obtain existence and uniqueness of~$S_{\rm can}^{\{-1\}}$ and~$Q_{\rm can}^{\{0\}}$ satisfying \eqref{Qq00}.
Moreover, we have $S_{\rm can}^{\{-1\}}\in\mathcal{A}_q\otimes \fn$
and $Q_{\rm can}^{\{0\}}\in \mathcal{A}_q\otimes V$. 
Next, comparing the terms with principal degree~$-k$, $k\geq1$, of both sides of \eqref{20230403:eq1} we obtain
\beq\label{Qcank}
Q_{\rm can}^{\{-k\}} \+ \Bigl[e,S_{\rm can}^{\{-k-1\}}\Bigr] \= 
M_k
\,,\eeq
where $M_k$ takes value in  $\mf b$ and depends polynomially in the entries of $q$ and~$S_{\rm can}^{\{-1\}},\dots,S_{\rm can}^{\{-k\}}$, and their derivatives, which by induction
have already been determined.
Using again the direct sum decomposition~\eqref{DSgaugeequal} 
we obtain existence and uniqueness of~$S_{\rm can}^{\{-k-1\}}$ and~$Q_{\rm can}^{\{k\}}$ satisfying \eqref{Qcank}.
Furthermore, by induction, from \eqref{Qcank} it follows that  $S_{\rm can}^{\{-k-1\}}\in\mathcal{A}_q\otimes \fn$
and $Q_{\rm can}^{\{-k\}}\in \mathcal{A}_q\otimes V$. 

Next, we show that all the entries of 
$Q_{\rm can}$ are gauge invariants.
By applying the differential algebra homomorphism $f$ to the entries of both sides of~\eqref{canform} and using
\eqref{eq:gauge} we get 
\[e^{\ad f(S_{\rm can})} (\p \+ \Lambda(\lambda) \+ Q)
\=e^{\ad f(S_{\rm can})} e^{\ad S} (\p \+ \Lambda(\lambda) \+ q) 
\=\p \+ \Lambda(\lambda) \+ f(Q_{\rm can}) \,.\]
By the Baker--Campell--Hausdorff formula, there exists~$B\in\mc A_{\bm q,\bm S}\otimes\mf n$ such that
$e^{\ad f(S_{\rm can})} e^{\ad S} =e^{\ad B}$. Hence, we get
\[ e^{\ad B} (\p \+ \Lambda(\lambda) \+ q) \= \p \+ \Lambda(\lambda)\+ f(Q_{\rm can})\,,\]
with $f(Q_{\rm can})\in\mc A_{\bm q,\bm S}\otimes V$. By uniqueness we have $(B,f(Q_{\rm can}))=(S_{\rm can},Q_{\rm can})$.
\end{proof}

It has been shown~\cite{DS}
that $V$ has the decomposition:
\begin{equation}\label{eq:Vdec}
V \=\bigoplus_{a=1}^{\ell} V_a \,, \quad \dim_{\CC} V_a\=1\,,
\end{equation}
where the 
non-zero elements in $V_a$ have principal degree $-m_a'$, $m'_a$ denoting the exponents of the semisimple Lie algebra $\mf a$.
(Note that $m_a'$ is also an exponent of $\widetilde{ \mf g}$ and there could be repetitions among the exponents of $\mf a$, namely $m_a'=m_b'$ for some $a,b=1,\dots\ell$.)

We take a basis $\bigl\{v_1,\dots, v_{\ell}\bigr\}$ of~$V$ with $\deg v_a= - m_a'$  
and write $Q_{\rm can}=\sum_{a=1}^{\ell} u_a v_a$. The next result is an immediate consequence of Proposition \ref{NQ}.
\begin{cor}[\cite{DS}]\label{ruu}
The elements $u_1,\dots,u_\ell\in\mc A_{\bm q}$ are independent, and the differential algebra of gauge invariants $\mathcal{R}\subset\mc A_{\bm q}$ is an algebra of differential polynomials (cf. \eqref{diffpoly}) in $u_1,\dots, u_\ell$.
\end{cor}
The next result has been proved in~\cite{BDY,DVY}.
\begin{lemma}[\cite{BDY,DVY}]\label{OmegaR}
For every $a,b=1,\dots, n$, and $k_1,k_2\geq0$, we have
\beq
\Omega_{a,k_1;b,k_2} \; \in \; \mathcal{R} \,.
\eeq
\end{lemma}
\begin{proof}
Applying the differential algebra $f$ to the defining equation~\eqref{tau-symm-eq} we have
\beq\label{s1}
\sum_{k_1,k_2\geq 0} \frac{f(\Omega_{a,k_1;b,k_2})}{\lambda^{k_1N_m+1} \mu^{k_2N_m+1}}
\=  \pi_{\lambda,\mu} \Biggl(\frac{\bigl(f(R_{m_a} (\lambda)) \, | \, f(R_{m_b} (\mu))\bigr)}{(\lambda-\mu)^2} \,-\,
\frac{\delta_{a+b,n+1}}r \frac{m_a  \lambda^{N_m} + m_b \mu^{N_m}}{(\lambda-\mu)^2}\Biggr)\,.
\eeq
Observe that for the basic resolvents $R_{m_a}(\lambda) \in  \mathcal{A}_{\bm q}\otimes  \g((\lambda^{-1}))$,
$a=1,\dots,n$,
the series $f(R_{m_a}(\lambda))$ satisfy the following three properties:
\begin{equation}
\begin{split}\label{20230403:eq2}
& \bigl[e^{\ad S} {\mathcal L}, f(R_{m_a}(\lambda))\bigr] \= 0 \, ,
\\
& f(R_{m_a}(\lambda)) \= \Lambda_{m_a}(\lambda) 
\+ \mbox{lower order terms w.r.t.} \deg^{\rm pr} \, , \\
& \bigl(f(R_{m_a}(\lambda))\, | \, f(R_{m_b}(\lambda)) \bigr) \= h \, \delta_{a+b,n+1} \, \lambda^{N_m} \, . 
\end{split}
\end{equation}
Using the facts that  $e^{\ad S}$ is a Lie algebra automorphism and that $\mf n\subset\widetilde{\mf g}^{<0}$, and the invariance of the bilinear form, from \eqref{basicdef1}-\eqref{boundaryr} it follows that
the series $e^{\ad{S}} R_{m_a}(\lambda)$ also satisfy~\eqref{20230403:eq2}.
Then, by the uniqueness of solutions to \eqref{20230403:eq2} (which can be proved similarly to the uniqueness of solutions to \eqref{basicdef1}-\eqref{boundaryr}) it follows that
\beq\label{s2}
f(R_{m_a}(\lambda)) \= e^{\ad{S}} R_{m_a}(\lambda)\,,
\quad a=1,\dots,n\,.
\eeq
Substituting~\eqref{s2} in~\eqref{s1}, using invariance of the bilinear form, and comparing the result with the RHS of~\eqref{tau-symm-eq}) we immediately get that
\[f(\Omega_{a,k_1;b,k_2}) \= \Omega_{a,k_1;b,k_2}\,,\quad a,b=1,\dots,n\,, k_1,k_2\geq0\,.\]
This concludes the proof.
\end{proof}

The next result was proven in \cite{BDY0} (cf. \cite{DLZ0}) in the case of an untwisted affine Kac-Moody algebra
$\widehat{\mf g}$ and the choice of the special vertex $c_0$ of its Dynkin diagram.
The proof uses the fact that the Lie algebra $\mf a$ defined in \eqref{dec:dynkin} is simple, which happens also for the choice of the special vertex $c_0$ in the twisted case.
Using similar arguments we then arrive at the following result.
\begin{prop}\label{prop:Miura}
For the pair $(c_0,\widehat{\mf g})$
the collection of gauge invariants
\begin{equation}\label{eq:Vmiura}
\bm V=(\Omega_{a,0;1,0})_{a=1}^\ell\in\mc R^\ell
\end{equation}
is a Miura-type $\ell$-ple.
\end{prop}
Using the results in 
\cite{DSKV13} on the structure of the series $H(\lambda)$  defined in
\eqref{eq:fundamental},
Proposition \ref{prop:Miura} can also be deduced from the relation showed in \cite{DVY} between the tau-structure $\Omega_{a,k_1;b,k_2}$ and 
$H(\lambda)$.

\subsection{DS hierarchies and modified DS hierarchies}\label{sec:5.4}

\subsubsection{DS derivations}

Recall the family of admissible derivations $D_{a,k}^{\rm pre}$ on $\mc A_{\bm q}$, $a=1,\dots,n$, $k\geq0$, defined in
\eqref{pre-DS}.
Let us extend
$D_{a,k}^{\rm pre}$ to an admissible derivation (which, by an abuse of notation, we still denote with the same symbol)
$D_{a,k}^{\rm pre}:\mc A_{\bm q,\bm S}\to \mc A_{\bm q,\bm S}$ by defining
\begin{equation}\label{eq:extension}
D_{a,k}^{\rm pre}(S_{i})=0\,,\quad i=1,\dots\dim\mf n
\,,
\end{equation}
and extending by the Leibniz rule.
The following result holds.
\begin{lemma}\label{fddf}
For every $a=1,\dots,n$ and $k\geq0$, we have
\[f \circ D_{a,k}^{\rm pre}(p)  \= D_{a,k}^{\rm pre} \circ f(p) \,, \]
for every $p\in\mc A_{\bm q}$.
\end{lemma}
\begin{proof}
It suffices to prove the claim for $p=q_i$, $i=1,\dots,\dim\mf b$, or, equivalently, for the element $q=\sum_{i=1}^{\dim \mf b}q_iv_i\in\mc A_{\bm q}\otimes\mf b$.
From the definition of the admissible derivations $D_{a,k}^{\rm pre}$ \eqref{pre-DS} and the Lax operator \eqref{eq:lax} we have
\begin{align*}
 D_{a,k}^{\rm pre} (q) &\= 
 \left[\left(\lambda^{kN_m} R_{m_a}(\lambda)\right)_+ \,,\, \L\right] 
 \\
& \= 
 \left[\left(\lambda^{kN_m} R_{m_a}(\lambda)\right)_+ \,,\, \Lambda(\lambda)+q\right] 
 \,-\, \p\left(\left(\lambda^{kN_m} R_{m_a}(\lambda)\right)_+\right)
 \,.
\end{align*}
By applying the differential algebra homomorphism $f$ to both sides of the above identity
and using equation \eqref{s2}
we obtain
\begin{align*}
 f(D_{a,k}^{\rm pre} (q))
 &\= 
 \left[\left(\lambda^{kN_m} e^{\ad{S}} R_{m_a}(\lambda)\right)_+ \,,\, \Lambda(\lambda)+Q\right] 
 \,-\, \p \left(\left(\lambda^{kN_m} e^{\ad{S}} R_{m_a}(\lambda)\right)_+\right)
 \\
&
\=
\left[\left(\lambda^{kN_m} e^{\ad{S}} R_{m_a}(\lambda)\right)_+ \,,\, \partial+\Lambda(\lambda)+Q\right] 
 \,.
 \end{align*}
On the other hand, using equations \eqref{eq:gauge}, \eqref{eq:extension} and \eqref{pre-DS} we have
\begin{align*}
& D_{a,k}^{\rm pre} (f(q)) \= D_{a,k}^{\rm pre} (Q)
\=D_{a,k}^{\rm pre} (\partial+\Lambda(\lambda)+Q)
\= D_{a,k}^{\rm pre} (e^{\ad S} \mathcal{L})
\= e^{\ad S}  D_{a,k}^{\rm pre} (\mathcal{L})\\
&\= e^{\ad S} \Bigl[\bigl(\lambda^{kN_m} R_{m_a}(\lambda)\bigr)_+ \,,\, \L \Bigr]
\=  \Bigl[e^{\ad S} \bigl(\lambda^{kN_m} R_{m_a}(\lambda)\bigr)_+ \,,\,  \p+ \Lambda(\lambda)+Q \Bigr]
\,.
\end{align*}
Recall that $S$ does not contain $\lambda$, therefore we conclude that 
\[f \bigl(D_{a,k}^{\rm pre} (q)\bigr) \= D_{a,k}^{\rm pre} (f(q))\,. \]
\end{proof}

\begin{cor}\label{fdd}
For every $a=1,\dots,n$ and $k\geq0$ we have
$$
D_{a,k}^{\rm pre}(\mc R)\subset\mc R\,.
$$
\end{cor}
\begin{proof}
Let $p\in\mc R$. By Lemma \ref{fddf} we have
\[f (D_{a,k}^{\rm pre}(p) ) \= D_{a,k}^{\rm pre}(f(p)) \= D_{a,k}^{\rm pre}(p)\,. \]
Hence, $D_{a,k}^{\rm pre}(p)\in\mc R$.
\end{proof}

\subsubsection{The DS hierarchy}

It follows from Corollary~\ref{fdd} that 
$D_{a,k}^{\rm pre}$, $a=1,\dots,n$, $k\geq 0$, induce admissible derivations 
$D_{a,k}=D_{a,k}^{\rm pre}|_{\mc R}:\mathcal{R}\rightarrow \mathcal{R}$. 
By applying $D_{a,k}$, $a=1,\dots, n$, $k\geq0$, to both sides of~\eqref{canform} we get
$$
D_{a,k}(Q_{\rm can})\=D_{a,k}^{\rm pre}\left(e^{\ad S_{\rm can}}\mc L\right)
\,.
$$
We use equation \eqref{pre-DS} and Lemma 4.3 in \cite{DVY} to rewrite the RHS above and get the following formula for the
action of the admissible derivation $D_{a,k}$ on $\mc R$:
\begin{equation}\label{eq:Dak}
D_{a,k}(Q_{\rm can})\=
\Biggl[ \Bigl(e^{\ad{S_{\rm can}}} \lambda^{k N_m}R_a(\lambda)\Bigr)_+ + 
\sum_{i\geq 0} \frac{1}{(i+1)!}(\ad S_{\rm can})^{i} \bigl(D_{a,k}^{\rm pre} (S_{\rm can})\bigr) \,,\, \mc L_{\rm can}\Biggr]
\,,
\end{equation}
where $\mc L_{\rm can}=\p +\Lambda(\lambda) + Q_{\rm can}$, where $Q_{\rm can}=\sum_{\alpha=1}^\ell u_{\alpha}v_\alpha$.
Comparing the coefficients of $v_\alpha$ in both sides of \eqref{eq:Dak} we get the corresponding hierarchy of evolution equations given by \eqref{pd}:
\beq\label{DShier}
\frac{\p u_\alpha}{\p t_{a,k} }\= D_{a,k} (u_\alpha)\,,\quad\alpha=1,\dots,\ell\,,a=1,\dots,n\,,k\geq0\,,
\eeq
which is known as the \emph{DS hierarchy of type $(X_n^{(r)},c_m)$}.

\begin{prop}\label{D10}
We have that $D_{1,0}=-\partial$.
\end{prop}

\begin{proof}
We start recalling from \eqref{basicdef2} and \eqref{dec:sigmatwisted} that
$$
R_1(\lambda)_+=\Lambda(\lambda)+\widetilde b\,,
$$
where $\widetilde b\in\mc A_{\bm q}\otimes (\mf a\cap \widetilde{\mf g}^{\leq 0})=\mc A_q\otimes \mf b$. Hence, using equations
\eqref{eq:Lambda}, the fact that $S_{\rm can}\in\mc A_{\bm q,\bm S}\otimes\mf n$ and the commutation relations \eqref{20171109:eq1} we have that
$$
\left(e^{\ad S_{\rm can}}R_1(\lambda)\right)_+=\Lambda( \lambda)+b\,,
$$
for some $b\in\mc A_q\otimes\mf b$. We can then rewrite 
\eqref{eq:Dak}, for $a=1$ and $k=0$, as
\begin{equation}\label{eq:D10}
D_{1,0}(Q_{\rm can})
\= \Biggl[ \Lambda(\lambda) +b + 
\sum_{i\geq 0} \frac{1}{(i+1)!}(\ad S_{\rm can})^{i} \bigl(D_{1,0}^{\rm pre} (S_{\rm can})\bigr) \,,\, \mc L_{\rm can}\Biggr]
\,.
\end{equation}
Note that $D_{1,0}(Q_{\rm can})\in\mc R\otimes V$, and that
$$
\sum_{i\geq 0} \frac{1}{(i+1)!}(\ad S_{\rm can})^{i} \bigl(D_{1,0}^{\rm pre} (S_{\rm can})\bigr)\in\mc A_{\bm q}\otimes\mf n
\,.
$$
Hence, the identity \eqref{eq:D10} has the form
\begin{equation}\label{to solve}
\psi=[\Lambda(\lambda)+b+\theta,\mc L_{\rm can}]
\,,
\end{equation}
where $\psi\in\mc R\otimes V$ and $\theta\in\mc A_{\bm q}\otimes\mf n$.

In order to conclude the proof that $D_{1,0}=-\partial$ we will show that equation \eqref{to solve},
viewing $\psi\in\mc A_{\bm q}\otimes V$ and $\theta\in\mc A_{\bm q}\otimes\mf n$
as the unknown indeterminates,
has a unique solution given by
$$
\psi=-\partial(Q_{\rm can})\,,\qquad
\theta=Q_{\rm can}-b\,.
$$ 
Recall from \eqref{eq:Lambda} that $\Lambda(\lambda)=\widetilde e_m \lambda+e$. By comparing powers of $\lambda$ in both sides of \eqref{to solve} we get
$$
0=[\widetilde e_m\lambda,Q_{\rm can}-b]
\,,
\quad
\psi=[e,Q_{\rm can}]
+[b+\theta,\partial+e+Q_{\rm can}]
\,.
$$
By the first equation we get that $Q_{\rm can}-b\in\mc A_{\bm q}\otimes\mf n$ (see \eqref{20171109:eq1}).
We rewrite the second equation above as
\begin{equation}\label{tosolve2}
\psi+[e,\theta]=
\phi+[\theta,\partial+Q_{\rm can}]\,,
\end{equation}
where $\phi=[e,Q_{\rm can}]+[b,\partial+e+Q_{\rm can}]$.
Recall the notation $A^{[k]}$ introduced in the proof of Proposition \ref{NQ}
to denote the homogeneous component of $A$ with principal degree $k$. Comparing coefficients of principal degree $k$, $k\leq0$,
in both sides of \eqref{tosolve2} we get
$$
\psi^{[k]}+[e,\theta^{[k-1]}]=
\phi^{[k]}-\partial(\theta^{[k]})
+\sum_{\substack{h,l\leq0\\ h+l=k}}[\theta^{[h]},Q_{\rm can}^{[l]}]\,.
$$
Note that $\theta^{[0]}=0$, hence for $k=0$ the above equation becomes
$$
\psi^{[0]}+[e,\theta^{[-1]}]=
\phi^{[0]}\,,
$$
which uniquely specifies $\psi^{[0]}$ and $\theta^{[-1]}$ due to the direct sum decomposition \eqref{DSgaugeequal}.
For $k \leq -1$, using induction and~\eqref{DSgaugeequal} we obtain at each step the existence and uniqueness of~$\psi^{[k]}$ and~$\theta^{[k-1]}$.
We conclude that the solution $(\psi,\theta)$ to~\eqref{tosolve2}, hence to \eqref{to solve},
exists and it is unique. We finally observe that
$\theta=Q_{\rm can}-b\in\mc A_{\bm q}\otimes\mf n$ and $\psi=-\partial(Q_{\rm can})\in\mc A_{\bm q}\otimes V$ clearly
solve \eqref{to solve}. By uniqueness it then follows that $D_{1,0}(Q_{\rm can})=-\partial(Q_{\rm can})$ thus concluding the proof.
\end{proof}

\subsubsection{The modified DS hierarchy}\label{sec:modDS}

Let us denote by $\mf h=\widetilde{\mf g}^0\subset \widetilde{\mf g}$ the subspace of zero principal degree \eqref{deg:principal}. By definition of the principal gradation, this
subspace is spanned by the Chevalley generators $h_0,\dots, h_\ell$. Recall that the Cartan subalgebra $\mf a^0\subset\mf a$ is spanned by the Chevalley generators
$h_0,\dots,h_{m-1},h_{m+1},\dots h_\ell$. Using the relation $\sum_{i=0}^\ell a_i^\vee h_i=0$ in $\widetilde{\mf g}$ , we have a Lie algebra isomorphism $\mf h\simeq \mf a^0$,
for every choice of the vertex $c_m$ of the Dynkin diagram of $\widetilde{\mf g}$.

It is proven in \cite{DSKV14} (see also \cite{DS}) that the canonical Lie algebra projection $\mf a\twoheadrightarrow\mf a^0\cong\mf h$ induces an injective differential ring homomorphism
\begin{equation}\label{eq:Miuramap}
\mc R\hookrightarrow\mc A(\mc O_{\bm v})
\,,
\end{equation}
where $\bm v=(v_1,\dots,v_\ell)$, and $\mc O_{\bm v}$ is an algebra of ordinary functions on the $\ell$-dimensional space $\mf h$.

It follows that the admissible derivations $D_{a,k}$ on $\mc R$ defined in \eqref{eq:Dak}, induce admissible derivations $D_{a,k}^{\rm mDS}$ on
$\mc A(O_{\bm v})$.
Hence, we have a hierarchy of evolution equations
\begin{equation}\label{eq:modDS}
\frac{\partial v_\alpha}{\partial t_{a,k}}=D_{a,k}^{\rm mDS}(v_\alpha)\,,
\quad\alpha=1,\dots,\ell\,,a=1,\dots,n\,,k\geq0\,,
\end{equation}
which is called the \emph{modified DS hierarchy}. This hierarchy is independent on the choice of the vertex $c_m$ since it is defined in terms of the principal gradation using the isomorphism $\mf h\cong \mf a^0$.

\subsection{A new proof of integrability of the DS hierarchy}\label{sec:int}

The following result was proved in~\cite{DS}, 
for which we outline an alternative proof applying Theorem~\ref{tauint}.
\begin{theorem}[\cite{DS}] \label{thm-main}
The family of the admissible derivations $D_{a,k}:\R \rightarrow \R$ $($defined by~\eqref{eq:Dak}$)$, $a=1,\dots,n$, $k\geq0$, is integrable.
\end{theorem}
\begin{proof}
Let $\Omega_{a,k_1;b,k_2}\in{\mathcal R}$, $a,b=1,\dots,n$, $k_1, k_2\geq 0$, be defined by~\eqref{tau-symm-eq}.
From Lemma~\ref{OmegaR} and equation~\eqref{symm2} we have that 
\[D_{a,k_1}(\Omega_{b,k_2;c,k_3})  \= D_{c,k_3}(\Omega_{a,k_1;b,k_2}) \]
for every $a,b,c=1,\dots, n$, $k_1,k_2,k_3\geq0$.
Hence, together with equation \eqref{symm1}, the elements $\Omega_{b,k_2;c,k_3}$, $b,c=1,\dots, n$, $k_2,k_3\geq0$
form a tau-structure associated to 
the family $D_{a,k}$, $a=1,\dots,n$, $k\geq0$ (the requirement that at least 
one of $\Omega_{b,k_2;c,k_3}$ is not a constant can be verified easily). 

For the choice of the special vertex $c_0$ of the Dynkin diagram of the Kac-Moody algebra $\widehat{\mf g}$
we have, by Proposition \ref{prop:Miura}, that the vector $\bm V$ defined by \eqref{eq:Vmiura} 
is a Miura-type $\ell$-ple.
Furthermore, by Proposition \ref{D10} we have that $D_{1,0}=-\partial$. This implies that $[D_{1,0},D_{a,k}]=0$, for every $a=1,\dots, n$ and $k\geq0$.
The claim of the theorem, in the case of the choice of the special vertex $c_0$, then follows by Theorem~\ref{tauint}. In particular, in view of the results of Section \ref{sec:modDS},
it follows that the family of the admissible derivations $D_{a,k}^{\rm mDS}$ is integrable as well.
Recall also that the modified DS hierarchy does not depend on the choice of the vertex of the Dynkin diagram $c_m$ and that the differential ring homomorphism \eqref{eq:Miuramap} is injective for any choice of $c_m$, thus this homomorphism can be inverted in view of the Dubrovin-Zhang theory \cite{DZ-norm}. The claim of the theorem then follow for any choice of the vertex $c_m$. 
\end{proof}

\begin{remark}
Already during the preparation of the paper~\cite{DVY} we wanted to give a unified proof of the 
Miura-type property of the vector ${\bf V}$ defined by~\eqref{eq:Vmiura}, as this would give a new way 
of writing the DS hierarchies (generalizing the new algorithm given in~\cite{BDY} for the untwisted case with the choice of the 
special vertex). We were able to do this in several new examples, but we still expect that it is true in general and that a unified proof could be found.
\end{remark}

As a consequence of Theorem \ref{thm-main} the DS hierarchy of type $(X_n^{(r)},c_m)$ is integrable and admits a tau-structure
$\Omega_{a,k;b,l}$, $a,b=1,\dots,n$, $k,l\geq0$, given by Definition \ref{tau-symm-def}. 
Hence, for suitable space of functions in space-time $\mc B$ described in Section \ref{sec:fun}, 
there exists a tau-function of an arbitrary solution to 
the hierarchy.

\begin{remark}
For the DS hierarchy of type $(X_n^{(1)}, c_0)$ one can construct the so-called \emph{topological solution}, see e.g. \cite{BDY, Dickey, DLZ0, LRZ}.
The corresponding topological tau-function has an important application in the quantum singularity theory 
(see \cite{BDY, DVY, LRZ} and the references therein).

For an arbitrary affine Kac-Moody algebra and a vertex of its Dynkin diagram, 
we define the {\it generalized Br\'ezin--Gross--Witten (BGW) solution} to the DS hierarchy~\eqref{DShier} 
 as the unique solution
$$u^{\rm gBGW}(x,\bm t)=(u_1^{\rm gBGW}(x,\bm t),\dots,u_\ell^{\rm gBGW}(x,\bm t)),\quad \bm t =(t_{a,k}\mid a=1,\dots,n\,,k\geq0),$$ 
 specified by the following initial data:
$$
u_\alpha^{\rm gBGW}(x,0) = \frac{C_\alpha}{(1-x)^{m_\alpha+1}}, \quad \alpha=1,\dots, \ell\,,
$$
where $C_\alpha$, $\alpha=1,\dots,\ell$, are arbitrarily given constants. 
We refer to the tau-function (cf. \eqref{eq:tau}) of the generalized BGW solution
$u^{\rm gBGW}(x,\bm t)$
as the {\it generalized BGW tau-function} for the DS 
hierarchy~\eqref{DShier}.
The generalized BGW tau-function for the KdV hierarchy was introduced in 
a special case by Br\'ezin--Gross~\cite{BG1} and Gross--Witten~\cite{GW}, 
and in general by Mironov--Morozov--Semenov~\cite{MMS} (see also~\cite{A}). 
For an untwisted affine Kac--Moody algebra and $c_0$ 
the above definition of the generalized BGW tau-function was given 
in joint work by Dubrovin, Zagier and the second-named author of the present paper.
\end{remark}

\begin{appendix}
\section{Discrete integrable systems and tau-structure}\label{App:A}
The notion of tau-structure generalizes to other classes of dynamical systems.
We consider in this appendix the case of differential-difference equations as an example.

In this framework, the algebra of differential polynomials $\mc A$ in the variables $\bm u=(u_1,\dots,u_\ell)$, should be replaced by the \emph{ring of difference polynomials} $(\mc A^{\textrm{discr}},S)$, that is the datum of a ring $\mc A^{\textrm{discr}}$ containing infinitely many variables $u_{\alpha,m}$, $\alpha=1,\dots,\ell\,,m\in\mb Z$, and an automorphism $S:\mc A^{\textrm{discr}}\to\mc A^{\textrm{discr}}$ (namely $S$ is invertible and $S(fg)=S(f)S(g)$, for every $f,g\in\mc A^{\textrm{discr}}$) such that $S(u_{\alpha,m})=u_{\alpha,m+1}$.

Note that, under some mild assumptions, the pair $(\mc A^{\textrm{discr}},S)$ can be embedded to~$(\widehat{\mathcal{A}},\p)$ via:
\[u_{\alpha,m}\to 
e^{\epsilon m\p}(u_\alpha)\,, \qquad S \rightarrow e^{\epsilon \p}\,, \]
where,
for an operator $X:\widehat{\mc A}\to\widehat{\mc A}$,
the formal operator $e^X$ is defined using the Taylor expansion of the exponential function $e^X= \sum_{k\geq0} \frac {X^k}{k!}$. Hence, the definition
of tau-structure given in Section \ref{section2} can be specialized to the pair $(\mc A^{\textrm{discr}},S)$. We give the details of this specialization in the remaining of this appendix.

An operator $D: \mc A^{\textrm{discr}}\rightarrow \mc A^{\textrm{discr}}$ is called 
a derivation if \[D(ab) \= D(a)  \, b \+ a \, D(b)\] for any $a,b\in\mc A^{\textrm{discr}}$. 
A derivation 
is called admissible, 
if $[D,S]=0$. Any $\ell$-ple $W=(W_1,\dots,W_\ell)\in(\mc A^{\textrm{discr}})^\ell$ defines
an admissible derivation~$D_W$ via $D_W(u_{\alpha,m}):=S^m(W_\alpha)$, 
$\alpha=1,\dots,\ell\,,m\geq 0$ and the Leibniz rule. 
Conversely, any admissible derivation~$D:\mc A^{\textrm{discr}}\to\mc A^{\textrm{discr}}$ comes 
from a unique~$\ell$-ple~$(D(u_1),\dots,D(u_\ell))\in(\mc A^{\textrm{discr}})^\ell$ (see \eqref{isoder}).

Let
$D_j$, $j\in E$, be a family of infinitely many linearly independent 
admissible derivations on~$\mc A^{\textrm{discr}}$.
The family $D_j$, $j\in E$, is called {\it integrable} if 
\[
[D_i,D_j]\=0\,,\quad \text{for every }i,j\in E\,.
\] 
The above notions extend, as done in Section \ref{section2}, to the case of derivations $D:\widehat{\mc A^{\textrm{discr}}}\to\widehat{\mc A^{\textrm{discr}}}$, where $\widehat{\mc A^{\textrm{discr}}}\subset\widehat{\mc A}$
is a certain completion of $\mc A^{\textrm{discr}}$.

We say that the family $D_j$, $j\in E$, {\it admits a tau-structure}
if there exist $\Omega_{i;j}\in \widehat{\mc A^{\textrm{discr}}}$, $i,j\in E$, such that
at least one of $\Omega_{i;j}$ is not a constant, and that for every $i,j,k\in E$ we have
\[\Omega_{i;j} \= \Omega_{j;i} \,, \quad 
D_i \, \bigl(\Omega_{j;k}) \= D_j \, \bigl(\Omega_{k;i}\bigr) \= D_k \bigl(\Omega_{i;j}\bigr) \, . \]

An $\ell$-ple 
\begin{equation}\label{eq:V}
\bm V=\bigl(V_1(\bm u,\bm u_{,\pm1},\bm u_{,\pm2},\dots;\epsilon),\dots,V_\ell(\bm u,\bm u_{,\pm1},\bm u_{,\pm2},\dots;\epsilon)\bigr)\in(\widehat{\mc A^{\textrm{discr}}})^\ell
\end{equation} 
is called {\it of Miura-type} if the condition in \eqref{Jacobiannondegen} holds.
Denote by $\mc A^{\textrm{discr}}_{\bm v}$ the ring of difference polynomials in the variables $\bm v=(v_1,\dots,v_\ell)$, and
let~$\bm V$ as in \eqref{eq:V} be a Miura-type $\ell$-ple.
We associate with 
this Miura-type $\ell$-ple a Miura-type map $\phi_{\bm V}: \widehat{\mc A^{\textrm{discr}}_{\bm v}}
\rightarrow\widehat{\mc A^{\textrm{discr}}}$
by assigning  on generators
\[
v_{\alpha}\in\widehat{\mc A^{\textrm{discr}}_{\bm v}}  \quad \mapsto \quad
\phi_{\bm V}(v_{\alpha})=V_{\alpha}(\bm u,\bm u_{1},\bm u_{2},\dots;\epsilon) \in\widehat{\mc A^{\textrm{discr}}}\,,
\]
and by extending it with the property that $\phi_{\bm V}\left(S (a)\right)=S\left(\phi_{\bm V}(a)\right)$, for every $a\in\widehat{\mc A^{\textrm{discr}}_{\bm v}}$.
As in Section \ref{section2}, it can be checked that for every $\alpha\in\{1,\dots,\ell\}$ there exists a unique element 
$U_\alpha(\bm v,\bm v_{\pm1},\bm v_{\pm2},\dots;\epsilon)\in\widehat{\mc A^{\textrm{discr}}_{\bm v}}$ such that
$$
\phi_{\bm V}(U_\alpha(\bm v,\bm v_{\pm1},\bm v_{\pm2},\dots;\epsilon))
\=u_\alpha \,,
$$
and that the element  $\bm U=(U_\alpha(\bm v,\bm v_{\pm1},\bm v_{\pm2},\dots;\epsilon))_{\alpha=1,\dots,\ell}\in(\widehat{\mc A^{\textrm{discr}}_{\bm v}})^{\ell}$
gives a ring homomorphism $\psi_{\bm U}: \widehat{\mc A^{\textrm{discr}}}\rightarrow \widehat{\mc A^{\textrm{discr}}_{\bm v}}$ (by assigning on generators
\[
u_{\alpha}\in\widehat{\mc A^{\textrm{discr}}}  \quad \mapsto \quad
\psi_{\bm U}(u_{\alpha})=U_{\alpha}(\bm v,\bm v_{1},\bm v_{2},\dots;\epsilon)\in\widehat{\mc A^{\textrm{discr}}_{\bm v}} \,,
\]
and requiring that $\psi_{\bm U}$ and $S$ commute), such that
$\phi_{\bm V}\circ \psi_{\bm U}={\rm id}_{\widehat{\mc A^{\textrm{discr}}}}$ and
$\psi_{\bm U}\circ \phi_{\bm V} = {\rm id}_{\widehat{\mc A^{\textrm{discr}}_{\bm v}}}$.
The homomorphism $\psi_{\bm U}$ is the inverse Miura-type map of~$\phi_{\bm V}$.
Proposition~\ref{tauint} holds verbatim in this framework. 

\end{appendix}

\end{document}